\documentclass[a4paper,12pt]{article}
\usepackage{amssymb,latexsym,amsmath}
\usepackage{theorem}
\usepackage{graphicx,epic,eepic}
\usepackage{color}

\theoremstyle{plain}%
    \newtheorem{theorem}{Theorem}[section]

    \newtheorem{lemma}[theorem]{Lemma}
    {
    \theorembodyfont{\upshape}
        \newtheorem{remark}[theorem]{Remark}
        \newtheorem{definition}[theorem]{Definition}
        
        \newtheorem{example}[theorem]{Example}
    }

\newenvironment{proof}
    {
    \noindent{\itshape Proof:}\rm \\ }
    {\\ \rightline{$\square$}}

 \newcommand{\N}{\mathbb{N}}

\begin{document}

\title{A mathematical model for networks with structures in the mesoscale}

\author{
    Regino Criado
    \thanks{e-mail: {\tt regino.criado@urjc.es}},
    Julio Flores
    \thanks{e-mail: {\tt julio.flores@urjc.es}},
    Alejandro Garc\'{\i}a del Amo
    \thanks{e-mail: {\tt alejandro.garciadelamo@urjc.es}}
    \\
    Jes\'{u}s G\'{o}mez-Garde\~{n}es
    \thanks{e-mail: {\tt jesus.gomez.gardenes@urjc.es}}\enspace and
    Miguel Romance
    \thanks{e-mail: {\tt miguel.romance@urjc.es}}
    \\ Department of Applied Mathematics\\
      Rey Juan Carlos University\\
      28933 M\'ostoles (Madrid), Spain
    }


\date{}

\maketitle

\begin{abstract}

The new concept of multilevel network is introduced in order to embody some topological properties of complex systems with  structures in the mesoscale which are not completely captured by the classical models.  This new model, which generalizes the hyper-network and hyper-structure models, fits perfectly with several real-life complex systems, including social and public transportation networks. We present an analysis of the structural properties of the multilevel network, including the clustering and the metric structures. Some analytical relationships amongst the efficiency and clustering coefficient of this new model and the corresponding parameters of the underlying network are obtained. Finally some random models for multilevel networks are given to illustrate how different multilevel structures can produce similar underlying networks and therefore that the mesoscale structure should be taken into account in many applications.

\end{abstract}

%
%

\section{Introduction}

During the last years the scientific community has shown that many relevant properties of communication systems, social networks and other biological and technological systems may be described in terms of complex network properties, including structural and dynamical properties and the interplay between them (see, for example, \cite{Albert}, \cite{Barabasi}, \cite{Boccaletti}, \cite{newman} or \cite{strogatz} and the references therein). The emergence of these appealing results has forged the Complex Networks Analysis as an attractive and multidisciplinary branch of research that includes topics of sociology (social networks), biological sciences (including metabolic and protein networks, genetic regulatory network and many others), neuro-sciences (neural interaction networks, cortical networks) engineering (phone call networks, computers in telecommunication networks, Internet, the World Wide Web), theoretical physics and applied mathematics.

The working objects of the Complex Networks Analysis are graphs (which are called {\sl complex networks})  which occasionally prove unable to capture the details present in some real-life problems. One example of this situation can be found in sociology. Social networks analysis is used in the social and behavioral sciences, as well as in economics, marketing, and industrial engineering (\cite{Wasserman}), but some questions related to the structure of social networks have been not understood properly. From the fact that a social network can be understood as a set of people or groups of people with some pattern of contacts or interactions between them (\cite{Scott},\cite{Wasserman}) - think of the {\sf Facebook} network - it might seem that all the connections or social relationships between the members of that network take place at the same level. But the real situation is far from this. The real relationships amongst the members of a social network take place (mostly) inside of different groups (levels) and therefore it cannot be modeled properly if only the natural local-scale point of view used in classic complex network models is taken into account. This phenomenon is related to the fact that several real-life complex systems have a community structure \cite{Girvan} that takes place in a scale that lies between the local and the global one and that gives information about the structure-functionality relationship. The problem of finding and analyzing the community structure of a complex network is one of the {\sl hot spots} of the complex network analysis, since these meso-scale structures reveal valuable information about roles of groups of nodes \cite{Newman2}.

The main goal of this paper is to present the concept of multilevel network as a new model useful to analyze with detail complex systems that has some structures in the mesoscale level. As we pointed out before, this kind of networks are present in many real-life situations, not only in social networks, but also in communication networks and all the other examples that exhibit a community structure. The introduction of a sharp mathematical model that fits this new structures may prove  crucial to properly analyze the dynamics that take place in these complex systems which are far from being completely understood. An example of these phenomena is the dissemination of culture in social networks by using the Axelrod Model that has been recently studied by using meso-scale networks \cite{Guerra}. Yet the study of other meso-scaled complex networks and their dynamics is one of the scientific challenges of the nowadays complex network analysis. There are other mathematical models that consider non-local structures such as the hypergraphs (or hypernetworks) \cite{Berge1} or the hyper-structures \cite{Criado} but they are not able to combine the local scale with the global and the mesoscale structures of the system.  For example, if we want to model how a rumor is spread within a social network, it is necessary to have in mind not only that different groups are linked only through some of their members but also that two people who know the same person don't have necessarily to know each other, and, in fact, these two people may belong to different groups or levels. If we try to model this situation with hyper-networks, then we will only be taking into account the social groups and not the actual relationships between their members; in contrast if we use the hyper-structure model, then we won't possibly  know to which social group each contact between two nodes belongs.

After multilevel networks are defined, we present a structural analysis of these  objects in terms of the clustering coefficient and the metric structure which must be properly defined in this setting. We will compare these parameters with the corresponding ones for the underlying projection networks and we will present the analytical relationships between them. Despite the fact that there are sharp connections between the structural properties of the multilevel network and the properties of their projection and layer (or slice) networks, we will present some examples that exhibit the deep differences between the classic complex network approach and this new point of view.

Finally, in section 3, we will give several growing models to produce multilevel random networks, inspired by the Barab\'asi-Albert preferential attachment model of complex networks \cite{Albert1}. These methods are inspired by the bipartite networks models introduced by Ramasco {\em et al.} \cite{ramasco} and by M.\,Peltom\"{a}ki and M.\,Alava \cite{peltomaki}, but in this case we have to take into account the meso-scale nature of the multilevel structure. The use of growing random models is natural since many of the  examples in real life that can be modeled by multilevel networks, such as social or transportation networks, are dynamic models whose structure grow in time by addition of new nodes, links and layers. By using these models we perform several random testings that show that we can get very different multilevel networks (i.e. with very different meso-scale structure) which still have similar local and global structures when considered as classic complex network. For example, we exhibit a multilevel network that at the meso-scale is like an  Erd\H{o}s-R\'enyi random network, but it has a  scale-free structure if we consider it as a complex network. These random testings reinforce the idea that in many cases we must take into account the meso-scale structure in order to get deep understanding of the structure and dynamics of many complex systems.

\section{Multilevel Networks: A Structural Analysis}

As we pointed out in the previous section, a sort of naive approach to complex systems with structures in the mesoscale could suggest that hyper-networks and hyper-structures fit perfectly to these real-life systems. The key-point that makes these mathematical model not to be the best solution for systems with structures in the mesoscale has to do with the fact that both hyper-networks and hyper-structures are node-based models, while many real systems combine a node-based point of view with a link-based perspective. For example, if we have a look again at a social network with structures in the mesoscale (such as several social groups within its structure), when we consider a relationship between two members of one social group(or several), we have to take into account not only the social groups that hold the members, but also the social group that holds the relationship itself, i.e., if, for example, there is a relationship between two people that share the same working group and the same sport group, we have to highlight if the relationship is due to sharing the same group of work or it has sport nature. This fact is not particular of social networks with structures in the mesoscale and a similar situation occurs, for example, in public transportation systems, where a link between two stations belonging to several transport lines can occur as a part of different lines.

In order to avoid this node-based nature of hyper-networks and hyper-structures, we propose to introduce the following concept that combines the node-based with the link-based perspective:

\begin{definition}
Let $G=(X,E)$ be a network. A {\sl multilevel network} is a triple $M=(X,E,\mathcal{S})$, where $\mathcal{S}=\{S_1,\ldots,S_p\}$ is a family of subgraphs $S_q=(X_q,E_q)$ of $G$ such that
\[
G=\bigcup_{q=1}^p S_q,
\]
i.e. $X=X_1\cup\cdots\cup X_p$ and $E=E_1\cup \cdots\cup E_p$. The network $G$ is the {\sl projection network} of $M$ and each subgraph $S_j\in \mathcal{S}$ is called a {\sl slice} of the multilevel network $M$.
\end{definition}

This mathematical model perfectly suits social systems (as well as other complex systems) with structures in the mesoscale, since each  social group can be understood as a {\sl slice graph} in a {\sl multilevel network}; thus we are simultaneously taking into account the nature of the links (i.e. relationships) and the nodes involved.

It is easy to check that this new mathematical object extends both the classic complex network model and also the hyper-network model \cite{Berge1}. Let us point it out very briefly. On the one hand, if $G=(X,E)$ is a network, then it can be easily seen as a multilevel network by considering $M=(X,E,\mathcal{S})$, where $\mathcal{S}=\{G\}$.

On the other hand if $\mathcal{H}=(X,H)$ is an hyper-network (i.e. $X$ is a non-empty set of nodes and $H=\{H_1,\ldots, H_p\}$ is a family of non-empty subsets of $X$, each of them called and {\sl hyper-link} of $\mathcal{H}$), then it can be seen as the multilevel set $M=(X,E_\mathcal{H},\mathcal{S}_\mathcal{H})$, given by
\[
\mathcal{S}_\mathcal{H}=\{K_{H_1},\ldots, K_{H_p}\},
\]
where $K_{H_j}$ is the complete network obtained by linking every pair of nodes of $H_j$ and
\[
E_\mathcal{H}=\bigcup_{r=1}^p K_{H_r}.
\]
By using similar arguments we can show that every hyper-structure \cite{Criado} can be understood as a particular multilevel network, by considering one slice network for each hyper-link in the hyper-structure and each slice graph being a set of isolated nodes.

Once we have introduced this new and novel mathematical object, we have to give suitable structural parameters to analyze it. We can give natural extensions of many of the usual tools of the complex networks' analysis, such as the clustering coefficient, an adjacency matrix/tensor, a natural network representation as a tripartite network or a geodesic structure, among many others.

\subsection{Clustering of multilevel networks}

In this subsection we extend the graph clustering coefficient introduced by Watts and Strogatz in \cite{wattsstro} to multilevel networks and we establish some relationships between the clustering coefficient of a multilevel network, the clustering coefficient of its slices and the clustering coefficient of its projection network. Recall that given a network $G = (X,E)$ the clustering coefficient of a given node $i$ is defined as
\[
c_G(i)=\frac{\text{number of links between the neighbors of $i$}}{\text{largest possible number of links between the neighbors
of $i$}} \, .
\]
We take $a/0$ to be $0$ for every number $a$. Thus, if we think of three people  $i$,$j$ and~$k$ with mutual relations between $i$ and $j$ as well as between $i$ and~$k$, the clustering coefficient of $i$  is supposed to represent the likeliness that $j$ and $k$ are also related. The clustering coefficient of $G$ is usually defined as the average of the clustering coefficients of all nodes.

In order to give a definition os the clustering coefficient of a node in a multilevel network $M=(X,E,\mathcal{S})$, where $\mathcal{S} = \lbrace S_1, \dots , S_p\rbrace$ and $S_q=(X_q,E_q)$ for every $1\le q\le p$, we need to introduce some notation. For every node $i\in X$ call $\mathcal{N}(i)$ the set of all neighbors of the node $i$ in the projection graph $G$.  For every $q=1, \dots , p$ we will denote $\mathcal{N}_q(i)=\mathcal{N}(i)\cap X_q$ and $\overline{G_q}(i)$ the subgraph of the slice $S_q$ generated by $\mathcal{N}_q(i)$, i.e. $\overline{G_q}(i)=(\mathcal{N}_q(i), \overline{E_q}(i))$, where
\[
\overline{E_q}(i)=\big\{\{k,j\}\in E_q \vert \, k,j\in \mathcal{N}_q(i)\big\}.
\]
Similarly we will define $\overline{G}(i)$ as  the subgraph of the projection network $G$ generated by $\mathcal{N}(i)$, i.e. $\overline{G}(i)=(\mathcal{N}(i), \overline{E}(i))$, where
\[
\overline{E}(i)=\big\{\{k,j\}\in E \vert \, k,j\in \mathcal{N}(i)\big\}.
\]
In addition the complete graph generated by $\mathcal{N}_q(i)$ will be denoted as usual by $K_{\mathcal{N}_q(i)}.$
Note that the largest possible number of links between the nodes of $\mathcal{N}_q(i)$ is precisely the number of links in $K_{\mathcal{N}_q(i)}$, i.e.
\[
|\overline{E_q}(i)|\le  \frac{|\mathcal{N}_q(i)|(|\mathcal{N}_q(i)|-1)}{2},
\]
where $|\mathcal{N}_q(i)|$ stands as above for the cardinality of $\mathcal{N}_q(i)$. With the previous notation we give the following

\begin{definition}
Let  $M=(X,E,\mathcal{S})$ be a  multilevel network. The {\sl clustering coefficient of a given node} $i$ in $M$ is defined as
\[
c_M(i)=\frac{\displaystyle 2\sum_{q=1}^p |\overline{E_q}(i)|}{\displaystyle\sum_{q=1}^p|\mathcal{N}_q(i)|(|\mathcal{N}_q(i)|-1)} \, .
\]
The {\sl clustering coefficient} of $M$ is the average of all $c_M(i)$ over the set of nodes.
\end{definition}

Notice that the clustering coefficient might have been defined differently. For instance we might have  considered the clustering coefficient of each $S_q$, say $c_q(i)$,  and then take the average over the set $\{S_q\}_{q=1}^p$. However we have opted for a definition that embodies the following idea: it is possible for a given node $i$ to have two neighbors $k,j$ with $\{i,k\}\in E_q$, $\{i,j\}\in E_{q'}$ with $q\neq q'$ and such that there is a link $\{j,k\}\in E_{q''}$ with $q''\neq q,q'$. This is a natural situation when we think of  social networks; indeed, one person $i$ knows $j$ from the aerobic class, $i$ also knows $k$ from her reading club while $j$ and $k$ know each other from the supermarket. This sort of situation produces more links in the subgraphs $\overline{G_q}(i)$ defined above than those already present in the slice $S_q$. Taking the alternative definition based on the clustering coefficients of the slices would not help to discriminate such situations; in contrast the definition proposed does. The following example should clarify this situation and justifies the restrictions of theorem \ref{multiclustering}.

\begin{example}\label{slide-ex}
If we consider the multilevel network $M=(X,E,\mathcal{S})$ where $X=\{1,2,3,4\}$, $S=\{S_1,S_2,S_3\}$ and $G$ is as in the next figure, it is easy to check that $c_G(i)=1$ for all $i \in X$ but however $c_{S_{q}}(i)=0$ for all $i \in X$ and $q=1,2,3$.
\begin{figure}[h!]
\begin{center}
\includegraphics[width=0.6\textwidth]{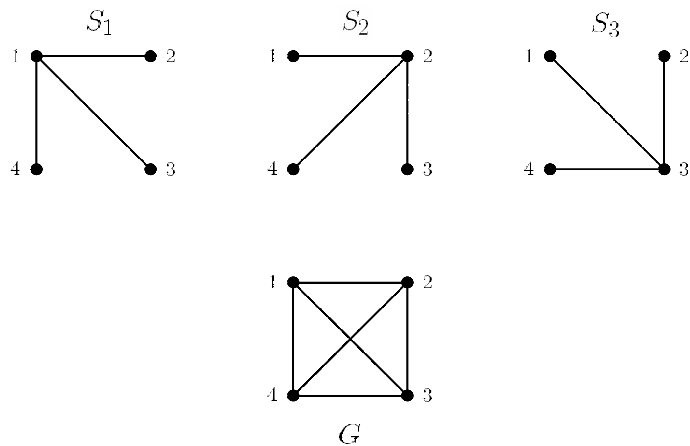}
\end{center}
\caption[]{An example of the extreme difference between clustering coefficients}
\end{figure}
\end{example}

As we have seen in the last example, the clustering coefficient of the projected graph $G$ and the clustering coefficient of each slice of $M$ may be very different. However  we are interested in establishing relations between the clustering coefficient of the multilevel network, $c_M(i)$, and the clustering coefficient in the projected network, $c_G(i)$, if possible. That is the target of the next theorem:

\begin{theorem}\label{multiclustering}
Let  $M=(X,E,\mathcal{S})$ be a  multilevel network and $i\in X$. If we consider
$A(i)=\{q\in\{1,\dots,p\} \, \vert \, |\mathcal{N}_q(i)|<2\}$ and we denote $\theta(i)=|A(i)|$, then
\begin{equation}\label{E:multiclustering}
\frac{1}{p-\theta(i)}\,c_G(i)\le c_M(i)\le p\,c_G(i)\left(1+(p-1)\left(4+ \frac {\theta(i)}{p-\theta(i)}\right)\right).
\end{equation}
Furthermore, if $M$ is homogeneous, then
\begin{equation}\label{E:multiclustering_homogeneo}
\frac{1}{p-\theta(i)}\, c_G(i)\le c_M(i)\le c_G(i).
\end{equation}
\end{theorem}

\begin{proof}
Let us start with the left-hand inequality in (\ref{E:multiclustering}). Using that $\mathcal{N}_q(i)\subseteq \mathcal{N}(i)$ for every $1\le q\le p$ and $\displaystyle \overline{E}(i)=\cup_{q=1}^p\overline{E_q}(i)$,  we get that
\[
\begin{split}
c_G(i)
&=\frac{2|\overline{E}(i)|} {|\mathcal{N}(i)|(|\mathcal{N}(i)|-1)}
 =\frac{\sum_{q=1}^p |\mathcal{N}_q(i)|(|\mathcal{N}_q(i)|-1)} {|\mathcal{N}(i)|(|\mathcal{N}(i)|-1)}
  \cdot\frac{2|\overline{E}(i)|} {\sum_{q=1}^p|\mathcal{N}_q(i)|(|\mathcal{N}_q(i)|-1)}\\
&\le (p-\theta(i))\,\frac{2|\overline{E}(i)|} {\sum_{q=1}^p|\mathcal{N}_q(i)|(|\mathcal{N}_q(i)|-1)}\le
(p-\theta(i))\, \frac{2\sum_{q=1}^p|\overline{E_q}(i)|}{\sum_{q=1}^p|\mathcal{N}_q(i)|(|\mathcal{N}_q(i)|-1)}\\
&=(p-\theta(i))\,c_M(i) \, ,
\end{split}
\]
since
\[
\sum_{q=1}^p |\mathcal{N}_q(i)|(|\mathcal{N}_q(i)|-1) =
\sum_{q \notin A(i)} |\mathcal{N}_q(i)|(|\mathcal{N}_q(i)|-1) \leq
(p-\theta(i))(|\mathcal{N}(i)|(|\mathcal{N}(i)|-1)) \, .
\]

In order to prove the remaining inequality notice that the number of links in the complete graph built up on the set of nodes $\mathcal N(i)$ never exceeds the sum over $q=1,\dots, p$ of the numbers of links of the complete graphs built up on the sets of nodes $\mathcal{N}_q(i)$ plus all the links possibly obtained by joining  nodes of $\mathcal{N}_k(i)$ to  nodes of~$\mathcal{N}_j(i)$ where $k$ and $j$ are chosen among  all possible pairs $(k,j)$ with $1\le k<j\le p$, i.e.
\begin{equation}\label{E:nodos_capa_proyeccion}
\frac{|\mathcal{N}(i)|(|\mathcal{N}(i)|-1)}{2}\le \frac 12\sum_{q=1}^p
|\mathcal{N}_q(i)|(|\mathcal{N}_q(i)|-1)+\sum_{k<j}|\mathcal{N}_k(i)||\mathcal{N}_j(i)|,
\end{equation}
where the last sum has $\binom{p}{2}$ terms. Notice also that we can always assume, by rearranging if necessary, that $|\mathcal{N}_q(i)|\le|\mathcal{N}_{q+1}(i)|$ for $1 \leq q \leq p - 1$. In addition to this, observe finally that
\begin{equation}\label{E:aristas_capa_proyeccion}
\begin{split}
\sum_{k<j}|\mathcal{N}_j(i)|\left(|\mathcal{N}_k(i)|-1\right)
&\le \sum_{q=1}^p\big[(|\mathcal{N}_q(i)|-1)|\mathcal{N}_q(i)|(q-1)\big]\\
&\le (p-1)\sum_{q=1}^p (|\mathcal{N}_q(i)|-1)|\mathcal{N}_q(i)|.
\end{split}
\end{equation}
Thus, since $\overline{E_q}(i)\subseteq \overline{E}(i)$ for every $1\le q\le p$ and by using (\ref{E:nodos_capa_proyeccion})
\begin{equation}\label{E:aristascapa_aristastotal}
\begin{split}
c_M(i)
 &=\frac{2\sum_{q=1}^p |\overline{E_q}(i)|} {\sum_{q=1}^p|\mathcal{N}_q(i)|(|\mathcal{N}_q(i)|-1)}
   \le \frac{2p|\overline{E}(i)|}{\sum_{q=1}^p|\mathcal{N}_q(i)|(|\mathcal{N}_q(i)|-1)}\\
 &=p\,c_G(i)\frac{|\mathcal{N}(i)|(|\mathcal{N}(i)|-1)}{\sum_{q=1}^p|\mathcal{N}_q(i)|(|\mathcal{N}_q(i)|-1)}\\
 &=p\,c_G(i)\left(
   1+\frac{2\sum_{k<j}|\mathcal{N}_k(i)||\mathcal{N}_j(i)|}{\sum_{q=1}^p|\mathcal{N}_q(i)|(|\mathcal{N}_q(i)|-1)}\right).
\end{split}
\end{equation}
Notice that if $1\le k<j\le p$, then
\[
|\mathcal{N}_k(i)||\mathcal{N}_j(i)|=\left(|\mathcal{N}_k(i)|-1+1\right)\,|\mathcal{N}_j(i)|\le \left(|\mathcal{N}_k(i)|-1\right)\,|\mathcal{N}_j(i)|+|\mathcal{N}_j(i)|.
\]
Now if we combine the last formula with (\ref{E:aristascapa_aristastotal}), we get that
\[
\begin{split}
c_M(i)
 &\leq p\,c_G(i)\left(
   1+\frac{2\sum_{k<j}|\mathcal{N}_k(i)||\mathcal{N}_j(i)|}{\sum_{q=1}^p|\mathcal{N}_q(i)|(|\mathcal{N}_q(i)|-1)}\right)\\
 &\le p\,c_G(i)\left(1+2(p-1)
   +2\frac{\sum_{q=1}^p(q-1)|\mathcal{N}_q(i)|}{\sum_{q=1}^p|\mathcal{N}_q(i)|\,\left(|\mathcal{N}_q(i)|-1\right)}\right)\\
 &\le p\,c_G(i)\left(1+2(p-1)
   +2\frac{(p-1)\sum_{q=1}^p|\mathcal{N}_q(i)|}{\sum_{q=1}^p|\mathcal{N}_q(i)|\,\left(|\mathcal{N}_q(i)|-1\right)}\right).
\end{split}
\]
It is easy to check that if $q\in A(i)$, then $|\mathcal{N}_q(i)|\,\left(|\mathcal{N}_q(i)|-1\right)=0$, which makes that
\[
\sum_{q=1}^p|\mathcal{N}_q(i)|\,\left(|\mathcal{N}_q(i)|-1\right)= \sum_{q\notin A(i)}|\mathcal{N}_q(i)|\,\left(|\mathcal{N}_q(i)|-1\right)
\ge \sum_{q\notin A(i)}|\mathcal{N}_q(i)|.
\]
Hence
\[
\begin{split}
\frac{\sum_{q=1}^p|\mathcal{N}_q(i)|}{\sum_{q=1}^p|\mathcal{N}_q(i)|\,\left(|\mathcal{N}_q(i)|-1\right)}
 &\le \frac{\sum_{q\notin A(i)}|\mathcal{N}_q(i)|+\sum_{q\in A(i)}|\mathcal{N}_q(i)|}{\sum_{q\notin A(i)}|\mathcal{N}_q(i)|}\\
 &\le 1+ \frac {\theta(i)}{2(p-\theta(i))},
\end{split}
\]
since if $q\in A(i)$, then $|\mathcal{N}_q(i)|\le 1$, while if  $q\notin A(i)$, then $|\mathcal{N}_q(i)|\ge 2$. Finally, if we use the last upper bound, we obtain that
\[
\begin{split}
c_M(i)&\le p\,c_G(i)\left(1+2(p-1)
   +2\frac{(p-1)\sum_{q=1}^p|\mathcal{N}_q(i)|}{\sum_{q=1}^p|\mathcal{N}_q(i)|\,\left(|\mathcal{N}_q(i)|-1\right)}\right)\\
&\le p\,c_G(i)\left(1+ (p-1)\left(4+ \frac {\theta(i)}{p-\theta(i)}\right)\right).
\end{split}
\]

On the other hand, if $M$ is homogeneous, then the same techniques as before shows that $c_G(i)\le (p-\theta(i))c_M(i)$. Since $M$ is homogeneous, then $\mathcal{N}_q(i)=\mathcal{N}(i)$ and $\overline{E_q}(i)\subseteq \overline{E}(i)$ for every $1\le q\le p$. Therefore we conclude that
\[
\begin{split}
c_M(i)
 &=\frac{2\sum_{q=1}^p |\overline{E_q}(i)|} {\sum_{q=1}^p|\mathcal{N}_q(i)|(|\mathcal{N}_q(i)|-1)}\\
 &\le \frac{2p|\overline{E}(i)|}{\sum_{q=1}^p|\mathcal{N}(i)|(|\mathcal{N}(i)|-1)}
   =\frac{2|\overline{E}(i)|}{|\mathcal{N}(i)|(|\mathcal{N}(i)|-1)}=c_G(i).
\end{split}
\]
\end{proof}

\begin{remark}
The definition of $\theta(i)$ shows that the estimation in the previous theorem  gets coarser as the number of slices for which $i$ has strictly less than two neighbors (possibly zero) increases; in a ``normal'' scenario this situation will not occur and the estimation will involve a polynomial of degree two in $p$. Anyway, since $0\le \theta(i)\le p$ and the function $f(x)=x/(p-x)$ is increasing in $[0,p)$, then the last result ensures that
\[
c_M(i)\le g(p)\,c_G(i),
\]
where $g(p)$ is of order $p^\alpha$ with $\alpha\in [2,3]$, but still it is not obvious whether the exponent 2 can be lowered.
\end{remark}

As it was said above there are other possible definitions for the clustering of a node in a multilevel network, based on the clusterings of that node in the slices, and then it is natural to have then related by means of estimations. However a remark is in order. Notice that it is possible to construct a network whose slices have clustering coefficients zero for every node while the clustering coefficient of the projection network is $1$ for every node. Indeed, if we consider the example \ref{slide-ex}  as our multilevel network $M$, then it is straightforward to check that for a given node~$i$, its clustering coefficient in each of the slices, $c_q(i)$, is zero, and hence every convex combination of them, while the clustering coefficient $c_M(i)$ is $1$. This means that we cannot expect to estimate $c_M(i)$ from convex combinations of $\{c_q(i)\}_{q=1}^p$. Still it makes sense to give a {\sl ``slicewise''} definition for the clustering coefficient of the node~$i$ in a multilevel network $M$ and find estimations by means of the clustering already defined.  Again some notation is required. Call
\[
\mathcal{N}_q^*(i)=\{j\in X \, \vert \, j\  \text{is a neighbor of}\  i\  \text{in}\ S_q\}
\]
and
${G_q}(i)=(\mathcal{N}_q^*(i),E_q(i))$, where
\[
E_q(i)=\left\{\{k,j\}\in E_q \, \vert \, k,j\in \mathcal{N}_q^*(i)\right\}
\]
Note that ${G_q}(i)$ is the subgraph of the slice $S_q$ generated by $\mathcal{N}_q^*(i)$. Thus the following definition is proposed.

\begin{definition}\label{slicewisedefinition}
Let $M=(X,E,\mathcal{S})$ be a  multilevel network and $i$ be a given node of $M$. The \textit{slice clustering} coefficient of $i$ is defined as
\[
c_G^{sl}(i)=\frac{2\displaystyle\sum_{q=1}^p |E_q(i)|}{\displaystyle\sum_{q=1}^p|\mathcal{N}_q^*(i)|(|\mathcal{N}_q^*(i)|-1)}.
\]
\end{definition}

Notice that ${G_q}(i)$ is a subgraph of $\overline{G_q}(i)$ as was illustrated in the example above. Accordingly $\mathcal{N}_q^*(i)\subseteq\mathcal{N}_q(i)$ and hence the  largest possible number of links between neighbours of $i$ in the $q$-slice, cannot exceed the corresponding largest possible number of links between neighbours of $i$ in $\mathcal{N}_q(i)$. As above we have the relation
\[
\frac{|\mathcal{N}(i)|(|\mathcal{N}(i)|-1)}{2}\le \frac 12\sum_{q=1}^p
|\mathcal{N}_q^*(i)|(|\mathcal{N}_q^*(i)|-1)+\sum_{k<j}|\mathcal{N}_k^*(i)||\mathcal{N}_j^*(i)|,
\]
where the last sum has $\binom{p}{2}$ terms.
Also we can assume, after rearrangement, that $|\mathcal{N}_q^*(i)|\le|\mathcal{N}_{q+1}^*(i)|$ for all $1\le q\le p$.
Thus by mimicking the proof of Theorem \ref{multiclustering}, we can get a relationship between the clustering coefficient in the projected network $c_G(i)$ and the slice clustering coefficient $c_G^{sl}(i)$.

\begin{theorem}
Let $M=(X,E,\mathcal{S})$ be a  multilevel network and $i$ be a given node of $M$.
Call $A(i)=\{q\in\{1,\dots,p\} \, \vert \, |\mathcal{N}_q^*(i)|<2\}$ and
$\theta(i)=|A(i)|$. Then
\[
c_G^{sl}(i)\le p\,c_G(i)\left(1+(p-1)\left(4+ \frac {\theta(i)}{p-\theta(i)}\right)\right).
\]
\end{theorem}

%
Notice that there are other possible variations for the definition of the slice clustering coefficient of a node $i$. For example, the average over the clustering coefficients $c_q(i)$ of the slides:
\[
c_G^{\overline{sl}}(i)=\frac{1}{p}\sum_{q=1}^p c_q(i) \, .
\]

The relation between this definition and definition \ref{slicewisedefinition}
\[
p \min_{1 \leq q \leq p} \frac{|\mathcal{N}_q^*(i)|(|\mathcal{N}_q^*(i)|-1)}{2} \, c_G^{\overline{sl}}(i) \leq c_G^{sl}(i) \leq
p \max_{1 \leq q \leq p} \frac{|\mathcal{N}_q^*(i)|(|\mathcal{N}_q^*(i)|-1)}{2} \, c_G^{\overline{sl}}(i)
\]
is easily derived from the following lemma.

\begin{lemma}
If $a_1,a_2,\ldots,a_n \in [0,\infty)$ and $b_1,b_2,\ldots,b_n \in (0,\infty)$, then
\[
\dfrac{n \min_{1\leq i \leq n} b_i}{\sum_{i = 1}^n b_i} \left( \dfrac{1}{n} \sum_{i = 1}^n \dfrac{a_i}{b_i} \right) \leq \dfrac{\sum_{i = 1}^n a_i}{\sum_{i = 1}^n b_i} \leq \dfrac{n \max_{1\leq i \leq n} b_i}{\sum_{i = 1}^n b_i} \left( \dfrac{1}{n} \sum_{i = 1}^n \dfrac{a_i}{b_i}\right) \, .
\]
\end{lemma}

\begin{proof}
It is straightforward since
\[
\frac {\sum_{i=1}^na_i}{\sum_{i=1}^nb_i}= \sum_{i=1}^n \left(\frac {b_i}{\sum_{j=1}^nb_j}\right)\frac {a_i}{b_i}.
\]
\end{proof}

The following examples illustrate the  behavior of the different concepts of clustering previously introduced. As it can be seen, if each one is separately evaluated one gets extremely different values.

\begin{example}
Consider
$M=(X,E,\mathcal{S})$ with $X = \lbrace 1,2,\ldots,n \rbrace$ and the slices
$S_1 = (\lbrace 1,2,3 \rbrace , \lbrace \lbrace 1,2 \rbrace, \lbrace 1,3 \rbrace, \lbrace 2,3\rbrace \rbrace)$  and $S_2 = ( \lbrace 1,2, \ldots ,n \rbrace, \lbrace \lbrace 1,2\rbrace,\lbrace 1,3\rbrace,\ldots,\lbrace 1,n\rbrace, \lbrace 2,3 \rbrace \rbrace )$. In this case
$$\displaystyle c_G^{\overline{sl}}(1) = \dfrac{1}{2}\left(\dfrac{1}{1} + \dfrac{1}{(n-1)(n-2)/2}\right) \sim \dfrac{1}{2}$$ while $$c_G^{	sl}(1) = \dfrac{1 + 1}{1 + (n-1)(n-2)/2} = O \left(\dfrac{1}{n^2}\right).$$ We see that the two local clustering coefficients  dramatically differ for $n \gg 1$.

On the other hand, consider $M=(X,E,\mathcal{S})$ with $X = \lbrace 1,2,\ldots,n^2 \rbrace$,
$\mathcal{S} = \lbrace S_1,S_2,\ldots,S_n \rbrace$, $S_j = (X,E_j)$, $E_1 = \lbrace \lbrace i,j \rbrace : 1 \leq i < j \leq n^2 \rbrace $, $E_j = \lbrace \lbrace 1,2 \rbrace, \ldots \lbrace 1,n \rbrace, \lbrace j,j + 1 \rbrace \rbrace$ if $2 \leq j \leq n - 1$, $E_n = \lbrace \lbrace 1,2 \rbrace, \ldots \lbrace 1,n \rbrace, \lbrace 2,n \rbrace \rbrace$. Then
$$c_G^{\overline{sl}}(1) = \dfrac{1}{n}\left(\dfrac{(n^2 - 1)(n^2 - 2)/2}{(n^2 - 1)(n^2 - 2)/2} + \dfrac{1}{(n-1)(n-2)/2} + \cdots + \dfrac{1}{(n-1)(n-2)/2}\right)$$
$$= O \left(\dfrac{1}{n^2}\right)$$
while
$$c_G^{sl}(1) = \dfrac{(n^2 - 1)(n^2 - 2)/2 + 1 + \ldots + 1}{(n^2 - 1)(n^2 - 2)/2 + (n-1)(n-2)/2 + \ldots + (n-1)(n-2)/2} \sim 1 \, .$$ Also in this case the two local clustering coefficients are very different  for $n \gg 1$.
\end{example}

\begin{example}
Consider $M=(X,E,\mathcal{S})$, where the slice $S_j = (X_j,E_j)$ is defined by $X_j = \lbrace 1,j+1,j+2 \rbrace$ and $E_j = \lbrace \lbrace 1,j+1 \rbrace,\lbrace 1,j+2 \rbrace,\lbrace j+1,j+2 \rbrace \rbrace$ for $1 \leq j \leq n-2$, $X_{n-1} = \lbrace 1,2,n \rbrace$, $E_{n-1} = \lbrace \lbrace 1,2\rbrace, \lbrace 1,n \rbrace, \lbrace 2,n \rbrace \rbrace$. Then $c_G^{\overline{sl}}(1) = 1$ while $c_G(1) = O (1/n)$.
\end{example}

\begin{example}
Now consider $M=(X,E,\mathcal{S})$ with $X_j = X = \lbrace 1,2, \ldots ,n \rbrace$ and $E_j = \lbrace \lbrace 1,k\rbrace : 2 \leq k \leq n\rbrace \cup \lbrace \lbrace j+1,k\rbrace : 2 \leq k \leq n, k \neq j+1\rbrace$ for $1 \leq j \leq n-1$. Then $c_G(1) = 1$ while $c_G^{\overline{sl}}(1) = O (1/n)$.
\end{example}

\begin{example}
Now consider $M=(X,E,\mathcal{S})$ with $X_j = X = \lbrace 1,2, \ldots ,n \rbrace$ and $E_j = \lbrace \lbrace 1,j+1\rbrace , \lbrace 1,j+2 \rbrace, \lbrace j+1,j+2 \rbrace \rbrace$ for $1 \leq j \leq n-2$, $E_{n-1} = \lbrace \lbrace 1,2 \rbrace , \lbrace 1,n \rbrace , \lbrace 2,n \rbrace \rbrace$. Then $c_G^{\overline{sl}}(1) = 1$ while $c_M(1) = O (1/n^2)$.
\end{example}

\begin{example}
Consider $M=(X,E,\mathcal{S})$ with $X_j = X = \lbrace 1,2, \ldots ,2n \rbrace $, $A = \lbrace \lbrace 1,k \rbrace : 2 \leq k \leq n \rbrace \cup \lbrace \lbrace k,l \rbrace : n+1 \leq k < l \leq 2n \rbrace$, $E_j = A \cup \lbrace \lbrace j+1,j+2 \rbrace \rbrace$ for $1 \leq j \leq n-2$, $E_{n-1} = A \cup \lbrace \lbrace n,2 \rbrace \rbrace $, $E_n = \lbrace \lbrace 1,k \rbrace : 2 \leq k \leq 2n \rbrace$. Then $c_M(1) \sim 1/4$ while $c_G^{\overline{sl}}(1) = O (1/n^2)$.
\end{example}

\begin{example}
Consider $M=(X,E,\mathcal{S})$ with $X_{i,j} = X = \lbrace 1,2, \ldots ,n \rbrace $ and $E_{i,j} = \lbrace \lbrace i,j \rbrace \rbrace$ for $1 \leq i < j \leq n$. Then $c_G^{\overline{sl}}(1) = 0$, $c_M(1) = O(1/n^2)$ and $c_G(1) = 1$.
\end{example}

\begin{example}
Consider $M=(X,E,\mathcal{S})$ with $X=\{1,2,\dots,n\}$, $X_{1}=\{1,2,3\}$, with $E_1=\{\{1,2\},\{1,3\},\{2,3\}\}$ and $X_j=\{1,j+2\}$ with $E_j=\{\{1,j+2\}\}$ for $j=2,\dots, n-2$.

Then $$c_G(1)=\frac{1}{(n-1)(n-2)}=O(1/n^2)$$
while  $$c_M(1)=\frac{1+0+\dots+0}{1+0+\dots+0}=1.$$
\end{example}

We conclude this section with a final observation. All what has been said refers to local clustering coefficients. The aggregation of the local quantities into a global one, or clustering of $M$, can naturally be considered. We left the details to the reader.

\subsection{Metric Structure: Distances and Efficiency}

If we want to introduce metric tools in a multilevel network $M=(X,E,\mathcal{S})$, we first have to give the notion of {\sl path} and {\sl length}. A {\sl path} $\omega$ in $M=(X,E,\mathcal{S})$ is a set of the form $\omega=\{(\ell_1,\cdots,\ell_q), (S_1,\cdots,S_q)\}$ such that
\begin{itemize}
 \item[{\it (i)}] $(\ell_1,\cdots,\ell_q)$ is a path in $(X,E)$,
 \item[{\it (ii)}] $(S_1,\cdots,S_q)$ is a sequence of slice graphs $S_1,\cdots,S_q\in \mathcal{S}$,
 \item[{\it (iii)}] For every $1\le j\le q$, we have that $\ell_j$ is an edge in the slice graph $S_j$.
\end{itemize}
By using this concept we can introduce a metric structure in a multilevel graph $M=(X,E,\mathcal{S})$ as follows.

\begin{definition}
Let $M=(X,E,\mathcal{S})$ be a multilevel network, $\beta\ge 0$ fixed and $\omega=\{(\ell_1,\cdots,\ell_q), (S_1,\cdots,S_q)\}$ be a path in $M$. The {\sl length} of $\omega$ is the nonnegative value
\[
 \ell(\omega)=q+\beta\sum_{j=2}^q \Delta(j),
\]
where
\[
\Delta(j)= \left\{
  \begin{array}{ll}
    1 & \mbox{if\enspace $S_j\ne S_{j-1}$,} \\
    0 & \mbox{otherwhise.}
  \end{array}
\right.
\]
The {\sl distance} in $M$ between two nodes  $i,j\in X$ is the minimal length among all possible paths in $M$ from $i$ to $j$.
\end{definition}

If we take $\beta=0$, the previous definition gives the natural metric in the projection graph, while if $\beta > 0$, we introduce new metrics that take into account, not only the global structure of the projection network, but also the interplay between the slice networks, that help to model the multi-scale nature of real-life social networks.

One way to calculate the distance matrix, that is a matrix $\Lambda=(\lambda_{ij})$ that contains in each position $\lambda_{ij}$ the distance in $M$ between vertices $i$ and $j$ in $X$, is to consider an auxiliary graph in the following way. Every vertex of the multilevel network $M$ is represented by a vertex in the auxiliary graph and, if a vertex in $M$ belongs to two or more slice graphs in $M$, then we duplicated it as many times as the number of slide graphs it belongs to. Every edge of $E$ is an edge in the auxiliary graph and there is one more (weighted) edge for each vertex duplication between the duplicated vertex and the original one. The distance between a duplicated vertex $j$ and another vertex $k$ is $d_{jk}=\min\{d_{j'k},d_{j''k}\}$, where $j'$ and $j''$ are the duplications of vertex $j$ in the auxiliary graph.

Once we have defined a distance in a multilevel network $M=(X,E,\mathcal{S})$, it is easy to generalize the concept of efficiency given in \cite{Criado} to this structure:

\begin{definition}
Let $M=(X,E,\mathcal{S})$ a multilevel network. The efficiency of $M$ is
$$
E(M)= \frac{1}{n(n-1)} \sum_{i,j \in X, i \ne j} \frac{1}{d_M(i,j)},
$$
where $n$ is the number of vertices in $X$ and $d_M(i,j)$ is the distance in $M$ between vertices $i$ and $j$.
\end{definition}

It is quite natural to try to establish comparisons between the efficiency $E(M)$ of a multilevel graph $M=(X,E,\mathcal{S})$ as it was just defined before and the efficiency  $E(G)$ of the underlying graph as a classic complex network, as it was introduced in \cite{latmar}. The actual analytical result that connects these two parameters are given in the following theorem.

\begin{theorem}
Let $M=(X,E,\mathcal{S})$ a multilevel network, and $G=(X,E)$. Then
\begin{equation}
\frac{1}{\beta+1} E(G)\leq E(M)\leq E(G),
\end{equation}
where $E(G)$ is the efficiency of the projection network $G$ and $E(M)$ is the efficiency of the multilevel network $M$.
\end{theorem}

\begin{proof}
Let $i,j \in X$ and  $\omega$ be a path in $M$ from $i$ to $j$. We denote by $d_M(i,j)$ the distance between  $i,j$ in the multilevel network $M$ and by $d_G(i,j)$ the distance between $i,j$ in the network $G$. In a similar way, we denote by $\ell_M(\omega)$ its length in the multilevel network $M$ and by $\ell_G(\omega)$ its length in the network $G$.

It is clear that $d_M(i,j)\geq d_G(i,j)$. On the other hand, since $\Delta(\cdot)\le 1$ then
\[
\begin{split}
d_M(i,j)&\leq \ell_M(\omega)= \ell_G(\omega)+ \beta \sum_{j=2}^{\ell_G(\omega)}\Delta(j)\\
&\leq \ell_G(\omega) + \beta (\ell_G(\omega) -1)=\ell_G(\omega)(\beta+1)-\beta\le \ell_G(\omega)(\beta+1),
\end{split}
\]
Hence, since it is clear that $d_G(i,j)=\min\{\ell_G(\omega)|\omega$ is a path in $M$ from $i$ to $j$\}, we get that
\[
\frac{1}{\beta+1} d_M(i,j)\leq d_G(i,j)
\]
and therefore, by combining the last expressions with the reverse bound obtained before we obtain that
\[
\frac{1}{\beta+1}\frac{1}{d_G(i,j)} \leq \frac{1}{d_M(i,j)}\leq \frac{1}{d_G(i,j)},
\]
which gives the result.
\end{proof}

\begin{remark}
The last bounds are sharp simply by making $\beta\longrightarrow 0^+$.
\end{remark}

We are equally interested in obtaining estimations of the efficiency of a multilevel graph in terms of the efficiencies of its slices. Precisely we obtain the following result.

\begin{theorem}
If $M=(X,E,\mathcal{S})$ is a multilevel network such that $S=\{S_1,\ldots, S_p\}$ and for every $1\le q\le p$ we denote $S_q=(X_q,E_q)$, then
\begin{equation}
\frac{|X_q|(|X_q|-1)}{|X|(|X|-1)} E(S_q)\leq E(M),
\end{equation}
where $|X_q|$ is the number of nodes of $S_q$ and $|X|$ is the number of  nodes of $M$ .
\end{theorem}

\begin{proof}
If we take $1\le q\le p$, since $X_q\subseteq X$ and $E_q\subseteq E$, then for every $i,j \in X_q$ and every path $\omega$ in $S_q$ from $i$ to $j$ in $S_q$, we get that $\ell_{S_q}(\omega)= \ell_M(\omega)\geq d_M(i,j)$. Then
$d_{S_q}(i,j) \geq d_M(i,j)$. Therefore
\[
\sum_{i,j \in X_q, i \ne j} \frac{1}{d_{S_q}(i,j)}\leq \sum_{i,j \in X_q, i \ne j} \frac{1}{d_{M}(i,j)}\leq \sum_{i,j \in X, i \ne j} \frac{1}{d_{M}(i,j)},
\]
which makes that
\[
\begin{split}
E(S_q)&=\frac{1}{|X_q|(|X_q|-1)}\sum_{i,j \in X_q, i \ne j} \frac{1}{d_{S_q}(i,j)}\\
&\leq\frac{1}{|X_q|(|X_q|-1)} \sum_{i,j \in X, i \ne j} \frac{1}{d_{M}(i,j)}=\frac{|X|(|X|-1)}{|X_q|(|X_q|-1)}E(M).
\end{split}
\]
\end{proof}

\begin{remark}
There is a reverse inequality that can be established when $M=(X,E,\mathcal{S})$ is an {\sl homogeneous multilevel network} (i.e., a multilevel network such that $X_q=X$ for every layer $S_q=(X_q,E_q)\in S$). Indeed, in this case the following relation is easily obtained
\begin{equation}
 E(S_q)\leq E(M)\leq E(G)\leq (|X|-1)E(S_q).
\end{equation}
\end{remark}

\section{Some models for multilevel random networks}

In this final section we present some models to produce multilevel random networks, inspired by the Barab\'asi-Albert preferential
attachment model of complex networks \cite{Albert1} and several bipartite networks models such as the collaboration network model
proposed by Ramasco {\em et al.} \cite{ramasco} or the sublineal preferential attachment bipartite model introduced by M.\,Peltom\"{a}ki and M.\,Alava \cite{peltomaki}. We propose growing random models since many of the real-life examples that can be modeled by multilevel networks, such as social or transportation networks, are dynamic models and their structure grow in time by addition of new nodes, links and layers. The presented randomized models are determined by the following rules:
\begin{itemize}
 \item[{\it (i)}] {\em Model parameters}. Our models have three main    parameters: $N$, $m$ and $p_{new}$. We set $N\in \N$ as the number of nodes in the multilevel network while $2\le m\le N$ account of  the number of nodes in each layer ({\em i.e.} if we take $m=2$, we recover the Barab\'asi-Albert model \cite{Albert1}). In our model $m$ will be fixed, but it can also be replaced by any other non-negative integer random variable if we want to produce more general models. Finally, we set $p_{new}\in(0,1]$ as the probability of joining a new node to the multilevel network.

\item[{\it (ii)}] {\em Initial conditions}. We start with a seed multilevel network made of one single layer $S_0$ of $m$ nodes that are linked all to all, ({\em i.e.} $S_0=K_m$). We can replace the {\sl all-to-all} structure by any other structure (such as a scale free or a Erd\H{o}s-R\'enyi network), but the results obtained are statistically equivalent.

 \item[{\it (iii)}] {\em Layer composition}. At each time step $t$, a new layer $S_t$ of $m$ nodes is added to the multilevel network. In order to determine $S_t$ we have to give its nodes and links. We start by choosing randomly an existing node of the multilevel network proportionally to its degree ({\em preferential election}). Therefore if at step $t-1$, the set of nodes of the multilevel network is $\{v_1,\ldots,v_n\}$, and $k_i$ denotes the degree of node $v_i$ at time $t-1$ in the projection network, then we choose the node $v_i$ randomly and independently with probability
     \[
     p_i=\frac {k_i}{\sum_{j=1}^{n} k_j}.
     \]
     The chosen node will be the first node of the new layer $S_t$ and we call it {\sl coordinator} of the layer. As an alternative, we could choose the coordinator proportionally to the number of layers that it belongs to but, in this case, the model will be similar to the collaboration network model proposed by Ramasco {em et al.} \cite{ramasco}. Each of the remaining $m-1$ nodes of $S_t$ will be a new node with probability $p_{new}$ and an existing node with probability $(1-p_{new})$. If we have to add an already existing node, we will uniformly and independently choose it at random. Note that we can replace the uniform random selection by other random procedure (such as preferential selection), but the random tests done suggest that the multilevel network obtained have statistically the same structural properties when $N$ is large enough ($N>10^3$). Anyway, we should have chosen $m$ nodes $\tilde v_1, \ldots, \tilde v_m$ ($\tilde v_1$ is the coordinator node) that will belong to the new layer $S_t$.

 \item[{\it (iv)}] {\em Layer inner-structure}. Once we have fixed the nodes $\tilde v_1, \ldots, \tilde v_m$ of the new layer $S_t$, we have to give its links. First, we link all the nodes to the coordinator node in order to warrant that the new layer is connected. In addition to this, we set new links between each pair of nodes, say $v_i$ and $v_j$ (with $1<i \ne j\leq m$) by using a random linking probability $p_{ij}$ that we will present later. At the end of this step we have defined completely the new layer $S_t$.

 \item[{\it (v)}] Finally, we repeat steps {\it (iii)} and {\it (iv)} until the number of nodes of the multilevel network is at least $N$.
\end{itemize}

These rules define a family of growing models for multilevel random networks that cover a wide range of different networks with different meso-scale structure, simply by changing the {\sl linking strategy}, $p_{ij}$, at each step. The simplest choice for this strategy is $p_{ij}=1$, {\sl i.e.} an all-to all strategy, thus linking all the nodes in the new slice. In this way, we recover a model similar to the Ramasco {\em et al.} model \cite{ramasco}, that can be used to produce hyper-networks with a meso-scale structure similar to that found in scientific collaboration systems.

In this paper we will focus on the following three different linking strategies:

\begin{enumerate}
 \item {\em Erd\H{o}s-R\'enyi type strategy ({\sl Model I})}. We fix a value $p_{ij}=p_{link}\in [0,1]$ so that we add each link $\{\tilde v_i,\tilde v_j\}$ randomly and independently with probability $p_{link}$, for every $2\le i\ne j\le m$. Note that this strategy does not take into account the meso-scale structure of the multilevel network, since the existence of the link $\{\tilde v_i,\tilde v_j\}$ is independent of the nodes $\tilde v_i,\tilde v_j$ in other slices of the multilevel network. Clearly, this is a {\sl toy model} inspired by the classic Erd\H{o}s-R\'enyi model.

\item {\em Assortative linking strategy ({\sl Model II})}. For every $2\le i\ne j\le m$, we add randomly the link $\{\tilde v_i,\tilde v_j\}$ proportionally to the number of common slices that hold simultaneously $\tilde v_i$ and $\tilde v_j$. Hence if we denote by $Q_{ij}$ the number of slices that hold simultaneously $\tilde v_i$ and $\tilde v_j$ at time step $t$ (including $S_t$) and by $q_i$ the number of slices that hold $\tilde v_i$ at time step $t$ (also including $S_t$), thus the probability of linking node $\tilde v_i$ with node $\tilde v_i$ is given by
     \[
     p_{ij}=\frac {2Q_{ij}}{q_i+q_j},
     \]
     for every $2\le i\ne j\le m$. The heuristic behind this strategy comes from social networks, since in this type of networks the relationships in a new social group are correlated with the previous relationships between the actors in other social groups \cite{Wasserman}. Hence, on the one hand, if two actors that belong to the new social group coincide in many (previous) social groups, then the probability of linking in this new social group is big and ,on the other hand, when two new actors join their first social group, the probability of establishing a relationship between them it is also high.

\begin{figure}[t!]
\begin{center}
\includegraphics[width=0.8\textwidth]{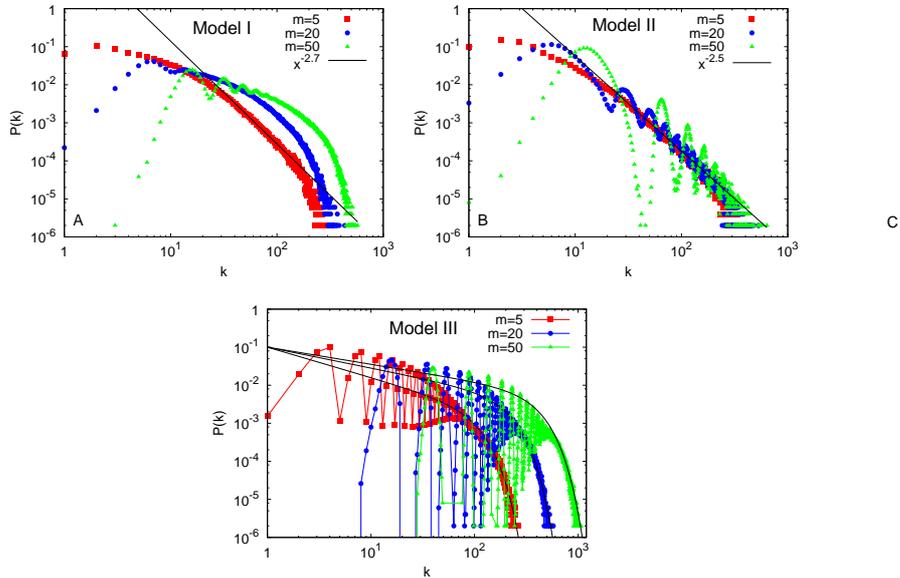}
\end{center}
\caption{Random testing of the degree distribution of the projection graph in the randomized multilevel models. Model I on the left, the model II is centered and Model III is on the right, each of them for $m=5,\,20,\,50$ nodes per slice. The distribution of each type of network and size of slice is computed for 100 random networks of 5000 nodes each one.}
\label{F:dist_grado}
\end{figure}

 \item {\em Disassortative linking strategy ({\sl Model III})}. If we take $2\le i\ne j\le m$, the probability of linking $\tilde v_i$ with $\tilde v_j$ in the slice $S_t$ is inversely proportional to the number of slices that hold simultaneously $\tilde v_i$ and $\tilde v_j$, i.e. if we denote by $p_{ij}$ such probability, then
     \[
     p_{ij}=1-\frac {2Q_{ij}}{q_i+q_j},
     \]
     by using the same notation as in the assortative strategy ({\sl Model II}). Now, the heuristic inspiring this strategy lies in some transportation networks, such as the airline networks, where the links in a new line try to connect nodes that are not connected in previous lines. By using the last expression, the newcomer nodes of the slice will be linked to old nodes that belong to many slices with high probability, while the newcomers will not be linked between them also with high probability. This strategy prevents the multiple reiteration of the same link in different slices and stresses the dissasortativity between newcomer nodes and nodes that belong to many slices.
\end{enumerate}
Once we have introduced the three linking strategies, we now present some numerical results to shows that, despite the fact that the above tree rules produce hyper-graphs with similar projection networks, the structures found in the mesoscale can be quite different.

\begin{figure}[t!]
\begin{center}
\includegraphics[width=0.7\textwidth]{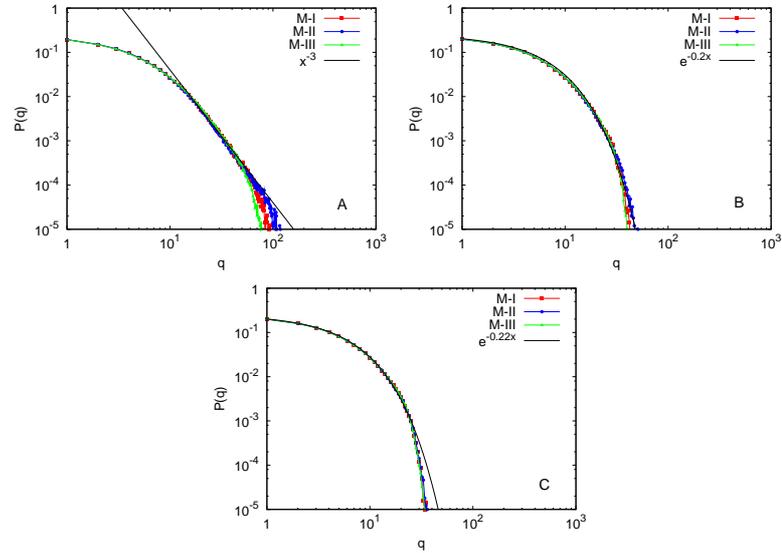}
\end{center}
\caption{Random testing of the hyper-degree distribution of the multilevel models (in $\log$ scale) for different number of nodes
  $m$ in each slice for 100 random networks of 5000 nodes each one. The case $m=5$ is on the left, $m=20$ is centered and $m=50$ is
  on the right, combined with some numerical approximation in each case.}
\label{F:dist_hiper}
\end{figure}

Let us start by showing in figure \ref{F:dist_grado} the degree distribution, $P(k)$, of the projection network.  It is clear that the degree distribution of the projection network obtained in three models display heavy-tailed profiles, and this fact is more relevant if the number of nodes of each slice is low, since in this case the Big Number Law's make that the degree distribution behaves as a power-law, $P(k)\sim k^{-\gamma}$ with $2\le \gamma\le 3$. Note that a discretizacion effect appears in all the distribution presented in figure~\ref{F:dist_grado} that makes that the degree distribution are oscillatory. This is due to the fact that each slice has the same number of nodes, and all of them are linked to the coordinator node, which makes that the degree of the nodes are concentrated in some values modulus $m$. In order to avoid this effect, we could modify step {\it (i)} replacing $m$ by a non-negative integer random variable that prevents that the degree of each node increases modulo $m$. This discretization effect is present in the three models and increases when the size of the slices is large, but it is more persistent in Model III, since in this case the dissasortative character of the model stresses this phenomenon.

In addition to the degree distribution of the projection network, in figure~\ref{F:dist_hiper} we show the distribution $P(q)$ of the number of slices that hold each node (we call it {\sl hyper-degree   distribution} for consistency with the hyper-graph theory notation). Obviously, from this figure we observe that the hyper-degree distribution is only depending on the number of nodes $m$ of each slide and it is not correlated with the linking model used to construct the network. On the other hand, the behavior of $P(q)$
depends quite strongly on the particular value of $m$. In particular, while $P(q)$ can be approximated by a truncated power law for $m=5$ (figure \ref{F:dist_hiper}.A), it presents an exponential decay in the cases $m=20$ and $m=50$ (figures \ref{F:dist_hiper}.B and \ref{F:dist_hiper}.C).

\begin{figure}[t!]
\begin{center}
\includegraphics[width=\textwidth]{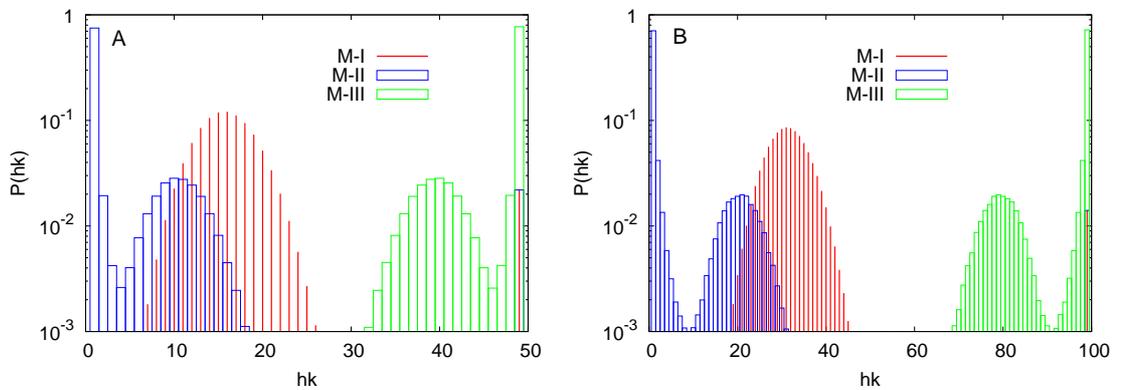}
\end{center}
\caption{Random testing of the average degree distribution of each slice for 100 random networks of 5000 nodes each one in the
  randomized multilevel models I, II and III. The case $m=50$ is in panel $A$ (on the left), while the case $m=100$ is in panel $B$ (on the right), both for models I, II and III.}
\label{F:histo}
\end{figure}

Figures \ref{F:dist_grado} and \ref{F:dist_hiper} has shown that either the structure of the projection network and hyper-degree
distribution are quite similar for the three random models presented. However, they are very different if we analyze their mesoscale structure, as figure \ref{F:histo} shows. This figure shows, for $m=50$ and $m=100$, the probability of finding a node connected to
$hk$ nodes within a randomly chosen slice of the network. It is straightforward to check that since the Erd\H{o}s-R\'enyi type
strategy is used in Model~I, then the resulting distribution $P(hk)$ is a binomial distribution with parameters $m$ and $p_{link}$
(asymptotically a Poisson distribution). However, note that in figure~\ref{F:dist_grado} we found that the degree distribution of the
projection network is scale-free. Correspondingly, Models~II and III display totally different mesoscale properties as shown by their
corresponding distributions $P(hk)$. On one hand in Model~II we observe that $P(hk)$ tend to accumulate around the maximum possible
value of $hk$, $m$, due to the assortative character of the linking strategy used. Conversely, in Model~III the association between the nodes belonging to the same slice is much lower due to the disassortativity of the model. Again, while both Models~II and III
gives two different bimodal averaged distributions $P(hk)$, the corresponding distribution for their projection networks are also of
scale-free type. Therefore, these results illustrate that many different multilevel networks (with different mesoscale structures)
can produce similar projection networks and therefore that we should take into account the mesoscale structure in order to give sharper models of many real-life problems.

\section*{Acknowledgements}

The authors of this work have been partially supported by the Spanish Government Project MTM2009-13848.

%
%



\end{document}